\documentclass[journal,9pt]{IEEEtran}

\usepackage{graphicx}
\usepackage{amsmath}
\usepackage{amssymb}
\usepackage{amsfonts}
\usepackage{color}
\usepackage{amsthm}
\usepackage{algorithm}
\usepackage{algpseudocode}
\usepackage{multirow}

\newtheorem{remark}{Remark}

\newtheorem{lemma}{Lemma}

\newtheorem{example}{Example}
\newtheorem{definition}{Definition}
\newtheorem{proposition}{Proposition}

\usepackage{comment}

\usepackage{graphicx}
\usepackage{amsmath}
\usepackage{amssymb}
\usepackage{amsfonts}
\usepackage{algorithm}
\usepackage{mathtools}
\usepackage{algpseudocode}
\newcommand{\set}[1]{\mathbb{#1}}
\newcommand{\N}{{\set{N}}}
\newcommand{\R}{{\set{R}}}
\newcommand{\E}{{{\cal E}}}
\newcommand{\V}{{{\cal V}}}
\newcommand{\G}{{{\cal G}}}
\renewcommand{\H}{{{\cal H}}}
\newcommand{\C}{{{\cal C}}}
\title{Distributed Nonlinear Conic Optimisation with partially separable Structure}

\author{Richard Heusdens  and Guoqiang Zhang

\thanks{R.\ Heusdens is with the Netherlands Defence Academy (NLDA), the Netherlands, and with the Faculty of Electrical Engineering, Mathematics and Computer Science, Delft University of Technology, Delft, the Netherlands (email: r.heusdens@\{mindef.nl,tudelft.nl\}).
}
\thanks{G.\ Zhang is with the University of Exeter, Exeter, United Kingdom (email: g.z.zhang@exeter.ac.uk)}
}

\begin{document}






\maketitle

\begin{abstract}
In this paper we consider the problem of distributed nonlinear optimisation of a separable convex cost function over a graph subject to cone constraints. We show how to generalise, using convex analysis, monotone operator theory and fixed-point theory, the primal-dual method of multipliers (PDMM), originally designed for equality constraint optimisation and recently extended to include linear inequality constraints, to accommodate for cone constraints. The resulting algorithm can be used to implement a variety of optimisation problems, including the important class of semidefinite programs with partially separable structure, in a fully distributed fashion. We derive update equations by applying the Peaceman-Rachford splitting algorithm to the monotonic inclusion related to the lifted dual problem. The cone constraints are implemented by a reflection method in the lifted dual domain where auxiliary variables are reflected with respect to the intersection of the polar cone and a subspace relating the dual and lifted dual domain. Convergence results for both synchronous and stochastic update schemes are provided and an application of the proposed algorithm is demonstrated to implement an approximate algorithm for maximum cut problems based on semidefinite programming in a fully distributed fashion.
\end{abstract}



\section{Introduction}

In the last decade, distributed optimisation \cite{boy:11} has drawn increasing attention due to the demand for either distributed signal processing or massive data processing over (large-scale) pear-to-pear networks of ubiquitous devices. 
Motivated by the increase in computational power of low cost micro-processors, the range of applications for these networks has grown rapidly. Examples include training a machine learning model, target localisation and tracking, healthcare monitoring, power grid management, and environmental sensing. Consequently, there is a desire to exploit the on-node computational capabilities of such networks to parallelise or even fully distribute computation.
In comparison to centralised counterparts, distributed networks offer several unique advantages, including robustness to node failure, scalability with network size and localised transmission requirements.

In general, the typical challenges faced by distributed optimisation over a network, in particular ad-hoc networks, are the lack of infrastructure, limited connectivity, scalability, data heterogeneity across the network,  data-privacy requirements, and heterogeneous computational resources \cite{dim:10,boy:11}. 
Various approaches have been developed to address one or more challenges within the network under consideration, depending on the applications involved. For example, \cite{boy:06,ust:10} introduced a pairwise gossip technique to enable asynchronous message exchange in the network, while \cite{ben:10b} discusses a hybrid approach combining gossip and geographic routing. In \cite{lut:13b}, a broadcast-based distributed consensus method was proposed to save communication energy. Alternatively, \cite{wai:05,sch:11} presents a belief propagation/message passing strategy, and \cite{lou:15,isu15} 
explores signal processing on graphs. Data-privacy requirements were addressed in \cite{li:20, li:21, Li:20sp, li:24}, while \cite{jon:18} considered data quantisation resulting in  communication efficient algorithms.

A method that holds particular significance in this study is to approach distributed signal processing by linking it with convex optimisation, as it has been demonstrated that numerous traditional signal processing problems can be reformulated into an equivalent convex form \cite{luo:06}. These methods model the problem at hand as a convex optimisation problem and solve the problem using standard solvers like dual ascent, the method of multipliers or ADMM \cite{boy:11} and PDMM \cite{zha:12,she:19}. While the solvers ADMM and PDMM may initially appear distinct due to their differing derivations, they are closely interconnected \cite{she:19}. The derivation of PDMM, however, directly leads to a distributed implementation where no direct collaboration is required between nodes during the computation of the updates. For this reason we will take the PDMM approach to derive update rules for distributed optimisation with general  cone constraints.

Let $G = ({\cal V},\mathcal{E})$ be an undirected graph
 where ${\cal V}$ is the set of vertices representing the nodes/agents/participants in the network, and $\mathcal{E}$ is the set of (undirected) edges, representing the communication links in the network. 
 PDMM  was originally designed to solve the following separable convex optimisation problem
\begin{align}
\begin{array}{ll} \text{minimise} & {\displaystyle \sum_{i\in {\cal V}} f_i(x_i)} \\\rule[4mm]{0mm}{0mm}
\text{subject to} & (\forall \{i,j\}\in \E) \,\,\, A_{ij}x_i + A_{ji}x_j = b_{ij}, \label{equ:linearCond}
\end{array}
\end{align}
in a synchronous setting. Recently, convergence result were presented for stochastic PDMM update schemes \cite{jor:23}, a general framework that can model variations such as asynchronous PDMM and PDMM with transmission losses.  In \cite{zha:22}, PDMM is modified for federated learning over a centralised network, where it is found that PDMM is closely related to the SCAFFOLD \cite{kar:20} and FedSplit \cite{pat:20} algorithm. Additionally, PDMM can be employed for privacy-preserving distributed optimisation, providing a level of privacy assurance, by utilising the fact that the (synchronous) PDMM updates take place within a particular subspace and the orthogonal complement can be used  to obscure local (private) data, a technique known as subspace perturbation \cite{li:20, li:21, Li:20sp, li:24}. Additionally, research in \cite{jon:18} demonstrates that PDMM exhibits robustness against data quantisation.
Recently, the PDMM algorithm has been extended to incorporate affine inequality constraints as well \cite{heu:24}.
This enhancement enables its application in solving linear programs in a distributed fashion.

Even though a large class of problems can be cast as a linear program, there is a growing interest in 
semidefinite programming, a relatively new field of optimisation. Semidefinite programming unifies several standard problems (e.g., linear and quadratic programming) and finds many applications in engineering and combinatorial optimisation.
In addition, semidefinite relaxation has been at the forefront of some highly promising advancements in the realms of signal processing, communications and smart grids. Its significance and relevance has been demonstrated across a variety of applications,  like sensor network localisation \cite{bis:06,so:07,sim:13}, multiple-input, multiple-output (MIMO) detection \cite{tan:01,ma:02,mob:07} and optimal power flow \cite{dal:13,liu:19}. The computational complexity associated with solving such problems typically grows unfavorably with the number of optimization variables (at least ${\cal O}(n^3)$) and/or the dimension of the semidefinite constraints involved. This poses limitations on the capability to solve large semidefinite programming instances. However, it is sometimes possible to reduce the complexity by exploiting possible structures in the problem such as sparsity or separability.
As a consequence, exploring distributed algorithms for semidefinite programming has received much research interest recently \cite{fuk:01,dal:13,sim:13,pak:18}.

\subsection{Main contribution}
In this paper we present a general framework for solving nonlinear convex optimisation problems with cone constraitns.  The framework has linear programming (LP), (convex) quadratic programming (QP), (convex) quadratically constrained quadratic programming (QCQP) and second-order cone programming (SOCP) as special cases. The resulting algorithms are fully distributed, in the sense that no direct collaboration is required between nodes during the computations, leading to an attractive (parallel) algorithm for optimisation in practical networks. To incorporate cone constraints, we impose polar cone constraints on the associated dual variables and then, inspired by \cite{she:19,heu:24}, derive closed-form update expressions for the dual variables via Peacheman-Rachford splitting of the monotonic inclusion related to the lifted dual problem. We perform a convergence analysis for both synchronous and stochastic update schemes, the latter based on stochastic coordinate descent.   

\subsection{Organisation of the paper}

The remainder of this paper is organized as follows. Section~\ref{sec:pre} introduces appropriate nomenclature and reviews  properties of monotone operators and operator splitting techniques. Section~\ref{sec:probl}  describes the problem formulation while Section~\ref{sec:operator} introduces a monotone operator derivation of PDMM with cone constraints. In Section~\ref{sec:convergence} we derive convergence results of the proposed algorithm and in Section~\ref{sec:stoch} we consider a stochastic updating scheme. Finally, Section~\ref{sec:num} describes applications and numerical experiments obtained by computer simulations to verify and substantiate the underlying claims of the document and the final conclusions are drawn in Section~\ref{sec:conclusion}.


\section{Preliminaries}
\label{sec:pre}
There exist many algorithms for iteratively minimising a convex function.
It is possible to derive and analyse many of these algorithms in a unified manner, using the abstraction of monotone operators.
In this section we will review some properties of monotone operators and operator splitting techniques that will be used throughout this manuscript. For a primer on monotone operator methods, the reader is referred to the self-contained introduction and tutorial \cite{ryu:16}. For an in-depth discussion on the topic the reader is referred to \cite{bau:17}.

\subsection{Notations and functional properties}

In this work we will denote by $\N$ the set of nonnegative integers, by $\R$ the set of real numbers, 
by $\R^n$ the set of real column vectors of length $n$, by $\R^{m\times n}$ the set of $m$ by $n$ real matrices, and by $S^n$ the set of $n$ by $n$ real symmetric matrices. We will denote by $\H$ a real Hilbert space with scalar (or inner) product $\langle \cdot \,,\cdot \rangle$ and induced norm $\|\cdot \|$; $(\forall x\in\H) \;  \|x\| = \sqrt{\langle x,x\rangle}$. The dimension of $\H$ is indicated by ${\rm dim}\,\H$.
The Hilbert direct sum of a family of real Hilbert spaces $(\H_i, \|\cdot\|_i)_{i\in{\cal I}}$, where $\cal I$ is a directed set, is the real Hilbert space $\bigoplus_{i\in{\cal I}}\H_i = \{ x = (x_i)_{i\in{\cal I}} \in \bigtimes_{i\in{\cal I}}\H_i \,|\, \sum_{i\in{\cal I}} \|x_i\|_i^2 < +\infty\}$, equipped with the addition $(x,y) \mapsto (x_i+y_i)_{i\in{\cal I}}$, scalar multiplication $(\alpha,x) \mapsto (\alpha x_i)_{i\in{\cal I}}$ and scalar  product $(x,y) \mapsto \sum_{i\in{\cal I}} \langle x_i,y_i\rangle_i$, where $\langle \cdot\,,\cdot \rangle_i$ denotes the scalar  product of $\H_i$. If $\H$ and $\G$ are real Hilbert spaces, we set $\C(\H,\G) = \{ T: \H\to\G \;|\; T \text{ is linear and continuous}\}$. Let $\R_{++} = \{c\in\R \,|\, c>0\}$.
A subset $K$ of $\H$ is a cone if $(\forall x\in K)(\forall c\in\R_{++})\, cx\in K$.  Important examples of cones are $\R^n_+  = \{ x\in\R^n \,|\, x_i \geq 0, \, i=1,\ldots,n\}$, the nonnegative orthant, and $S^n_+$, the positive semidefinite cone in $\R^{n\times n}$. 
If $\H=\bigoplus_{i\in{\cal I}}\H_i$ and $K_i$ is a cone in $\H_i$, then $\bigtimes_{i\in{\cal I}}K_i$ is a cone in $\H$.
The dual cone of $K$ is given by $K^* = \{x\in\H \,|\, (\forall y\in K) \; \langle x,y\rangle \geq 0\}$ and the polar cone of $K$ is $K^\circ = -K^*$.  Let $K$ be a nonempty subset of $\H$, and let $x\in\H$. We denote by $P_Kx$ the projection of $x$ onto $K$: $(\forall x\in\H)(\forall y\in K) \,\, \langle y-P_Kx,x- P_Kx\rangle \leq 0$. If $K$ is a nonempty closed convex cone in $\H$, then $x\in\H$ admits the conic decomposition $x = P_Kx + P_{K^\circ}x$ where $P_{K}x \perp P_{K^\circ}x$ \cite{mor:62}.
We call $K$ self-dual if $K^*=K$. Both $\R^n_+$ and $S^n_+$ are self-dual. 
Let $K$ be a proper\footnote{A cone $K\subseteq \H$ is called proper if it is closed, convex, solid and $K \cap (-K) = \{0\}$ ($K$ is pointed).} cone in $\H$. We associate with a proper cone $K$ a partial ordering on $\H$ defined by $(\forall x\in\H)(\forall y\in\H)\;\; x \preceq_K \!y \Leftrightarrow y-x\in K$. We also write $x \succeq_K \!y$ for $y \preceq_K \!x$.
As an example, when $K=\R_+$, the partial ordering $\preceq_K$ is the usual ordering $\leq$ on $\R$. If $K=S^n_+$,  the partial ordering $\preceq_K$ is the usual matrix inequality: $X \preceq_K \! Y$ means $Y-X$ is positive semidefinite. 

When $x$ is updated iteratively, we write $x^{(k)}$ to indicate the update of $x$ at the $k$th iteration. When we consider $x^{(k)}$ as a realisation of a random variable, the corresponding random variable will be denoted by $X^{(k)}$ (corresponding capital). The expectation operator is denoted by $\mathbb{E}$. 

Let ${\cal X,Y}$ be nonempty sets, and let $2^{\cal Y}$ be the power set of $\cal Y$, i.e., the family of all subsets of $\cal Y$. A set valued operator $T:{\cal X}\to 2^{\cal Y}$ is defined by its graph ${\rm gra}\,T = \{ (x,y)\in {\cal X}\times{\cal Y} \,|\, y \in Tx\}$. The domain of $T$ is ${\rm dom}\,T= \{x\in {\cal X} \,|\, Tx \neq \emptyset\}$. The kernel and range of $T$ are ${\rm ker}\,T = \{x\in {\cal X} \,|\, Tx =0\}$ and ${\rm ran}\,T = T(\cal X)$, respectively. The identity operator on $\H$ is denoted by $I$. If $Tx$ is a singleton or empty for any $x$, then $T$ is a function or single-valued,  usually denoted by $f$. 
The notion of the inverse of $T$, denoted by $T^{-1}$, is also defined through its graph, ${\rm gra}\,T^{-1} = \{ (y,x)\in {\cal Y}\times{\cal X} \,|\, y \in Tx\}$. Let $c\in\R_{++}$. We denote by $J_{cT} = (I + cT)^{-1}$ the resolvent of an operator $T$ and by $R_{cT} = 2J_{cT} - I$ the  reflected resolvent, sometimes referred to as the Cayley operator.  
The composition of two operators $T_1:{\cal X}\to 2^{\cal Y}$ and $T_2:{\cal Y}\to 2^{\cal Z}$ is given by $T_2\circ T_1 : {\cal X}\to 2^{\cal Z}$.  
The set of fixed points of $T$ is denoted by ${\rm fix}\,T = \{ x\in {\cal X} \,|\, Tx=x\}$.

Functional transforms make it possible to investigate problems from a different perspective and sometimes simplify the analysis. In convex analysis, a suitable transform is the Legendre transform, which maps a function to its Fenchel conjugate. We will denote by $\Gamma_0(\H)$ the set of all closed, proper, and convex (CCP) functions $f : \H \to \R \cup \{+\infty\}$.
The Fenchel conjugate of a function $f \in\Gamma_0(\H)$ is defined as $f^{\ast}(y) = \sup_x \left( \langle y,x \rangle -f(x) \right)$. The function $f$ and its conjugate $f^*$ are related by the Fenchel-Young inequality $f(x) + f^*(y) \geq \langle y,x \rangle$ \cite[Proposition 13.15]{bau:17}.   We will denote  by $\partial f$ the subdifferential of $f$. If $f\in \Gamma_0(\H)$, then $f=f^{**}$. Moreover, we have $y \in\partial f(x) \Leftrightarrow x\in \partial f^*(y) \Leftrightarrow f(x) + f^*(y) = \langle y,x \rangle$. 
If $f\in \Gamma_0(\H)$, the proximity operator ${\rm prox}_{cf}$ is defined as ${\rm prox}_{cf}\, x = {\arg}\min_{u\in\H}\left( f(u) + \frac{1}{2c}\|x-u\|^2\right)$ and is related to the resolvent of $\partial f$ by ${\rm prox}_{cf}\,x = J_{c\partial f}x$ \cite[Proposition 16.44]{bau:17}. If $\iota_{K}$ is the indicator function on a closed convex subset $K$ of $\H$, then ${\rm prox}_{\iota_{K}}\,x =  P_K x$.

An undirected graph is denoted by $G=(\mathcal{V},\mathcal{E})$, where $\mathcal{V}$ is the set of vertices representing the  nodes in the network and $\mathcal{E}=\{\{i,j\}\,|\, i, j\in \mathcal{V}\}$ is the set of undirected edges (unordered pairs) in the graph representing the communication links in the network. We use $\bar{\E} = \{(i,j)\,|\, i, j\in \mathcal{V}\}$ to denote the set of all directed edges (ordered pairs). Therefore, $|\bar{\E}|=2|\E|$. 
We use $\mathcal{N}_i$ to denote the set of all neighbouring nodes of node $i$, i.e., $\mathcal{N}_i=\{j\in\V \,|\,\{i,j\}\in \mathcal{E}\}$, and ${\rm deg}\,i = |{\cal N}_i|$ to denote the degree of vertex $i\in\V$.

\subsection{Monotone operators and operator splitting}

The theory of monotone set-valued operators plays a central role in deriving iterative convex optimisation algorithms. Global minimisers of proper functions can be characterised by the principle $\arg\!\min f \Leftrightarrow \{x\in\H \,|\, 0 \in \partial f(x)\}$. The subdifferential of a convex function
is a (maximally) monotone operator, and the problem at hand can thus be expressed as finding a zero of a monotone operator (monotone inclusion problem) which, in turn, is transformed into finding a fixed point of its associated resolvent.
The fixed point is then found by the fixed point (Banach-Picard) iteration, yielding an algorithm for the original problem. 
In this section we give background information about monotone operators and operator splitting to support the remainder of the manuscript.
\begin{definition}[Monotone operator] 
Let $T:\H\to 2^{\H}$. Then $T$ is monotone if 
\[
(\forall(x,u)\in{\rm gra}\,T)(\forall(y,v)\in{\rm gra}\,T) \quad \langle  v-u, y-x \rangle \geq 0.
\] 
The operator is said to be strictly monotone if strict inequality holds.
The operator is said to be uniformly monotone with modulus $\phi : \R_+\to [0,+\infty)$ if $\phi$ is increasing, vanishes only at 0, and
\[
(\forall(x,u)\in{\rm gra}\,T)(\forall(y,v)\in{\rm gra}\,T) \quad \langle  v-u, y-x \rangle \geq \phi( \|y-x\|).
\]
The operator is said to be strongly monotone with constant $m>0$, or $m$-strongly monotone, if $T-mI$ is monotone, i.e.,
\[
(\forall(x,u)\in{\rm gra}\,T)(\forall(y,v)\in{\rm gra}\,T) \quad \langle v-u, y-x \rangle \geq m \|y-x\|^2.
\]
The operator is said to be maximally monotone if for every $(x,u)\in \H\times \H$,
\[
(x,u)\in{\rm gra}\,T \,\, \Leftrightarrow \,\, \big(\forall (y,v)\in{\rm gra}\,T\big) \quad \langle  v-u, y-x \rangle \geq 0.
\]
In other words, there exists no monotone operator $S:\H\to 2^{\H}$ such that ${\rm gra}\,S$ properly contains ${\rm gra}\,T$.
\end{definition}
It is clear that strong monotonicity implies uniform monotonicity, which itself implies strict monotonicity.
\begin{definition}[Nonexpansiveness]
Let $D$ be a nonempty subset of $\H$ and let $T:D\to\H$. Then $T$ is nonexpansive if 
\[
(\forall x\in D)(\forall y\in D) \quad \|Ty-Tx\| \leq \|y-x\|.
\]
$T$ is called strictly nonexpansive, or contractive, if strict inequality holds. The operator is firmly nonexpansive if
\[
(\forall x\in D)(\forall y\in D) \quad \|Ty- Tx\|^2 \leq \langle Ty-Tx, y-x\rangle.
\]
\end{definition}
\begin{definition}[Averaged nonexpansive operator]
Let $D$ be a nonempty subset of $\H$ and let $T:D\to\H$ be nonexpansive, and let $\alpha \in (0,1)$. Then $T$ is averaged with constant $\alpha$, or $\alpha$-averaged, if there exists a nonexpansive operator $S : D \to \H$ such that $T = (1-\alpha)I + \alpha S$. 
\end{definition}
It can be shown that if $T$ is maximally monotone, then the resolvent $J_{cT}$ is firmly nonexpansive \cite[Proposition 23.8]{bau:17} and the reflected resolvent $R_{cT} = 2J_{cT}-I$ is nonexpansive \cite[Corollary 23.11\,(ii)]{bau:17}. We have
\[
0\in Tx \Leftrightarrow x \in (I+cT)x \Leftrightarrow (I+cT)^{-1}x \ni x \Leftrightarrow x=J_{cT}x,
\] 
where the last relation holds since $J_{cT}$ is single valued. Therefore, we conclude that a monotone inclusion problem is equivalent to finding a fixed point of its associated resolvent. Moreover, since $J_{cT} = \frac{1}{2}(R_{cT}+I)$ is $1/2$-averaged, we have, by the Krasnosel'skii-Mann algorithm, that the sequence generated by the Banach-Picard iteration $x^{(k+1)} = J_{cT}x^{(k)}$ is Fej\'{e}r monotone \cite[Definition 5.1]{bau:17} and converges weakly\footnote{For finite-dimensional Hilbert spaces weak convergence implies strong convergence.} to a fixed point $x^*$  of $J_{cT}$ for any $x^{(0)}\in {\rm dom}\,J_{cT}$ \cite[Theorem 5.15]{bau:17}, and thus to a zero of $T$.
In that case the Banach-Picard iteration $x^{(k+1)} = J_{c\partial f}x^{(k)}$ results in the well known proximal point method \cite[Theorem 23.41]{bau:17}.

For many maximal monotone operators $T$, the inversion operation needed to evaluate the resolvent may be prohibitively difficult. A more widely applicable alternative is to devise an operator splitting algorithm in which $T$ is decomposed as $T=T_1 + T_2$, and the (maximal monotone) operators $T_1$ and $T_2$ are employed in separate steps. Examples of popular splitting algorithms are the forward-backward method, Tseng's method,  the Peaceman-Rachford splitting algorithm and the Douglas-Rachford splitting algorithm, where the first two methods require $T_1$ (or $T_2$) to be single valued (for example the gradient of a differentiable convex function). The Peaceman-Rachford splitting algorithm is given by the iterates  \cite[Proposition 26.13]{bau:17}
\begin{align}
x^{(k)} &= J_{cT_1}z^{(k)}, \nonumber \\
v^{(k)} &= J_{cT_2}(2x^{(k)} - z^{(k)}), \label{eq:prs} \\
z^{(k+1)} &= z^{(k)} + 2(v^{(k)} - x^{(k)}). \nonumber
\end{align}
When $T_1$ is uniformly monotone, $x^{(k)}$ converges strongly to $x^*$ (notation $x^{(k)} \to x^*$) for any $z^{(0)}\in\H$, where $x^*$ is the solution to the monotonic inclusion problem $0 \in T_1x + T_2x$.
The iterates \eqref{eq:prs} can be compactly expressed using reflected resolvents as
\begin{align*}
x^{(k)} &= J_{cT_1}z^{(k)},\\
z^{(k+1)} &= (R_{cT_2}\circ R_{cT_1}) z^{(k)}.
\end{align*}
If either $R_{cT_1}$ or $R_{cT_2}$ 
is contractive, then $R_{cT_2}\circ R_{cT_1}$ is contractive and the Peacman-Rachford iterates converge geometrically. 
Note that since $R_{cT_2}\circ R_{cT_1}$ is nonexpansive, without the additional requirement of $T_1$ being uniformly monotone, there is no guarantee that the iterates will converge. In order to ensure convergence without imposing conditions like uniform monotonicity, we can average the nonexpansive operator. In the case of $1/2$-averaging, the $z$-update is given by
\[
z^{(k+1)}  = \frac{1}{2} \left(I + R_{cT_2}\circ R_{cT_1}\right) z^{(k)},
\]
which is called the Douglas-Rachford splitting algorithm.   
This method was first presented in \cite{dou:56, lio:79} and converges under more or less the most general possible conditions.
A well known instance of the Douglas-Rachford splitting algorithm is the alternating direction method of multipliers (ADMM) 
 \cite{gab:76,eck:93,boy:11} 
 or the split Bregman method \cite{bre:67}.

\section{Problem Setting}
\label{sec:probl}

Let $G = (\V,\E)$ be an undirected graph. We will consider the minimisation of a separable function over the graph $G$ subject to cone constraints. Let $N=|\V|$, $(n_1,\ldots,n_N)\in\N^N_+$ be strictly positive integers, and set $(\forall  i\in{\cal V}) \,\,{\cal L}_i = \{1,\ldots,n_i\}$. We consider the problem
\begin{equation}
\begin{array}{lll} \text{minimise} & {\displaystyle \sum_{i\in {\cal V}} f_i(x_i)} \\\rule[4mm]{0mm}{0mm}
\text{subject to} & (\forall  \{i,j\}\in {\cal E})  &A_{ij}x_i + A_{ji}x_j  - b_{ij} \in K_{ij}, \\
&(\forall  i\in{\cal V})(\forall \ell_i \in{\cal L}_i)  &h_{\ell_i} (x_i) \leq 0, \rule[4mm]{0mm}{0mm}
\end{array}
\label{eq:primalij}
\end{equation}
where $f_i \in \Gamma_0(\H_i)$, $A_{ij} \in\C(\H_i,\G_{ij})$, $b_{ij}\in\G_{ij}$, $K_{ij}\subseteq \G_{ij}$ is a closed convex cone, and $h_{\ell_i}\in \Gamma_0(\H_i)$. 
Note that \eqref{eq:primalij} also includes equality constraints of the form $A_{ij}x_i + A_{ji}x_j  = b_{ij}$ by setting $K_{ij} = \{0\}$, the trivial cone. The constraints $h_i(x_i)\leq 0$ can be used, for example, to set constraints on the individual entries of $x_i$ in the case $\H_i = \R^{m\times n}$, the Hilbert space of $m\times n$ real matrices. See \ref{sec:mimo} for a prototypical example.  Hence, problem \eqref{eq:primalij} describes the optimisation of a nonlinear convex function subject to cone constraints. We will refer to solving such programs as nonlinear conic programming (NLCP). If the objective function is linear and $K_{ij} = S^{n_{ij}}_+$, where $n_{ij} = {\rm dim}\, \G_{ij}$, \eqref{eq:primalij} reduces to semidefinite programming (SDP). In fact, \eqref{eq:primalij} has linear programming (LP), (convex) quadratic programming (QP), (convex) quadratically constrained quadratic programming (QCQP) and second-order cone programming (SOCP) as special cases.

In order to compactly express \eqref{eq:primalij}, we introduce
$x = (x_i)_{i\in\V}\in\H$, where $\H = \bigoplus_{i\in\V}\H_i$. Similarly, we define $\G = \bigoplus_{(i.j)\in\E}\G_{ij}$. Let $A : \H\to\G : (x_i)_{i\in\V} \mapsto (A_{ij}x_i+A_{ji}x_j)_{\{i,j\}\in\E}$, and $b = (b_{ij})_{\{i,j\}\in\E}$. With this, \eqref{eq:primalij} can be expressed as
\begin{equation}
\begin{array}{ll} \text{minimise} & f(x) \\\rule[4mm]{0mm}{0mm}
\text{subject to} & Ax-b  \in K, 
\end{array}
\label{eq:primal}
\end{equation} 
where  $f(x) = \sum_{i\in {\cal V}} \left( f_i(x_i) + \iota_{N_i}(x_i)\right)$, and $\iota_{N_i}$ is the indicator function on the set $N_i = \{x_i\in\H_i \,|\, (\forall {\ell_i} \in {\cal L}_i) \,\, h_{\ell_i} (x_i)\leq 0\}$, and $K = \bigtimes_{\{i,j\}\in{\cal E}}K_{ij}$. 
Hence, the inequality constraints $h_{\ell_i}(x_i)\leq 0$ are incorporated in the objective function using indicator functions for reasons that will become clear later in Section~\ref{sec:ineq}. The Lagrange dual function $g:\G\to \R\cup\{+\infty\}$ is defined as
\begin{align*}
g(\lambda) &= \inf_{x\in\H}\big( f(x) + \langle\lambda, Ax-b\rangle \big) \\
    &= -f^*(-A^*\lambda) - \langle\lambda,b\rangle, \rule[3mm]{0mm}{0mm}
\end{align*} 
with $f^*$ the Fenchel conjugate of $f$ and $A^*\in\C(\G,\H)$ is the adjoint of $A$.
With this, the dual problem is given by 
\cite[Proposition 27.17]{bau:17}
\begin{equation}
\begin{array}{ll}  \text{minimise} & f^*(-A^*\lambda) + \langle\lambda,b\rangle \\\rule[4mm]{0mm}{0mm}
\text{subject to} & \lambda \in K^\circ,
\end{array}
\label{eq:dual}
\end{equation}
where $K^\circ$ denotes the polar cone of $K$.  
We have $\lambda = (\lambda_{ij})_{\{i,j\}\in\E}$, where $\lambda_{ij}\in\G_{ij}$ denotes the Lagrange multipliers associated with the constraints on edge $\{i,j\}\in \cal E$. At this point we would like to highlight that in the case we have only equality constraints in \eqref{eq:primal}, that is, constraints of the form $Ax=b$, we have $K=\{0\}$ and $K^\circ=K^*=\H$, and the dual problem is simply an unconstrained optimisation problem. That is, problem \eqref{eq:dual} reduces to
\begin{equation}
\begin{array}{ll}  \text{minimise} & f^*(-A^*\lambda) + \langle\lambda,b\rangle.
\end{array}
\label{eq:dualeq}
\end{equation}

\section{Operater splitting of the lifted dual function}
\label{sec:operator}

Let $\tilde{f}_i(x_i) = f_i(x_i)+\iota_{N_i}(x_i)$.
Since $f$ is separable, we have
\[
f(x) = \sum_{i\in {\cal V}} \tilde{f}_i(x_i)  \quad\Leftrightarrow\quad f^*(y) = \sum_{i\in {\cal V}} \tilde{f}_i^*(y_i),
\]
that is, the conjugate function of a separable CCP function is itself separable and CCP. Moreover,
the adjoint operator $A^*$ satisfies
\begin{align*}
    \langle A^*\lambda,x\rangle &= \langle \lambda,Ax\rangle \\
    &= \sum_{\{i,j\}\in\E} \langle \lambda_{ij}, A_{ij}x_i + A_{ji}x_j \rangle \\
    &= \sum_{i\in\V} \sum_{j\in{\cal N}_i} \langle\lambda_{ij}, A_{ij}x_i \rangle \\
    &= \sum_{i\in\V} \langle\sum_{j\in{\cal N}_i}  A_{ij}^*\lambda_{ij}, x_i\rangle,
\end{align*}
from which we conclude that
\[
A^* : \G\to\H : (\lambda_{ij})_{\{i,j\}\in\E} \mapsto \left(\sum_{j\in{\cal N}_i} A_{ij}^* \lambda_{ij} \right)_{\!\!i\in\V}\!\!\!,
\]
and thus
\begin{equation}
f^*(-A^*\lambda) =  \sum_{i\in {\cal V}} \tilde{f}^*_i\bigg(\!-\!\sum_{j\in {\cal N}_i}A_{ij}^*\lambda_{ij}\bigg).
\label{eq:optconstrij}
\end{equation}
By inspection of \eqref{eq:optconstrij} we conclude that every $\lambda_{ij}$,
associated with edge $\{i,j\}$, is used by two conjugate  functions: $\tilde{f}_i^*$ and $\tilde{f}_j^*$. As a consequence, all conjugate functions depend on each other. 
We therefore introduce auxiliary variables to decouple the node dependencies. That is,  
we introduce for each edge $\{i,j\}\in \cal E$ {two} auxiliary {node} variables $\lambda_{i|j}$ and $\lambda_{j|i}$, one for each node $i$ and $j$, respectively, and require that  $\lambda_{i|j} = \lambda_{j|i}$. 
That is, let $\bar{\lambda} = (\lambda_{i|j})_{(i,j)\in\bar{\E}}\in\bar{\G}$ where $\bar{\G} = \bigoplus_{(i,j)\in\bar{\E}}\G_{ij}$, and define
\[
C : \H\to\bar{\G} : (x_i)_{i\in\V} \mapsto (A_{ij}x_i)_{(i,j)\in\bar{\E}},
\]
a permutation operator $P : \bar{\G}\to \bar{\G} : (\lambda_{i|j})_{(i,j)\in\bar{\E}} \mapsto (\lambda_{j|i})_{(i,j)\in\bar{\E}}$, and $d = (\frac{1}{2}b_{ij})_{(i,j)\in\bar{\E}}$. With this, we can reformulate the dual problem \eqref{eq:dual} as 
\begin{equation}
\begin{array}{ll} \text{minimise} \; \;&f^*(-C^*\bar{\lambda}) + \langle \bar{\lambda}, d\rangle  \\
\rule[4mm]{0mm}{0mm}
\text{subject to} & \bar{\lambda}\in \bar{K}^\circ, \\
& \rule[4mm]{0mm}{0mm}\bar{\lambda} = P\bar{\lambda},
\end{array}
\label{eq:exdual}
\end{equation}
where $C^*$ is the adjoint of $C$ and $\bar{K}^\circ = \bigtimes_{(i,j)\in{\bar{\E}}}K_{ij}^\circ$.  We will refer to \eqref{eq:exdual} as the lifted dual problem of the primal problem \eqref{eq:primal}. Note that  $(I+P)C : \H\to\bar{\G} : (x_i)_{i\in\V} \mapsto (A_{ij}x_i + A_{ji}x_j)_{(i,j)\in\bar{\E}}$ so that 
\begin{equation}
(I+P)Cx -d  \in \bar{K} \; \Leftrightarrow \; Ax- b \in K.
\label{eq:CxAx}
\end{equation}
Moreover, if $y\in \bar{\G}$ such that $y\in \bar{K}$, then $Py \in \bar{K}$ by construction.
Let $M = \{\bar{\lambda}\in\bar{\G} \,|\, \bar{\lambda}\in \bar{K}^\circ, \bar{\lambda}=P\bar{\lambda}\}$. With this, problem \eqref{eq:exdual} can be equivalently expressed as 
\begin{equation}
 \text{minimise} \; \; f^*(-C^*\bar{\lambda}) + \langle \bar{\lambda}, d\rangle  +  \iota_{M}(\bar{\lambda}). 
\label{eq:newprob}
\end{equation}
Again, by comparing general cone vs.\ equality constraint optimisation, the difference is in the definition of the set $M$; for equality constraint optimisation the set $M$ reduces to $M = \{\bar{\lambda}\in\bar{\G} \,|\, \bar{\lambda}=P\bar{\lambda} \}$.
The optimality condition for problem \eqref{eq:newprob} is given by the inclusion problem
\begin{equation}
0\in -C \partial f^*(-C^*\bar{\lambda}) + d + \partial \iota_{M}(\bar{\lambda}).
\label{eq:optconstr}
\end{equation}

In order to apply Peaceman-Rachford splitting to \eqref{eq:optconstr}, we  define  the operators $T_1$ and $T_2$ as $T_1 = -C \partial f^*(-C^*(\cdot)) + d$ and $T_2 = \partial \iota_{M}$. Since $T_1$ is the subdifferential of $f^*(-C^*\bar{\lambda}) + \langle \bar{\lambda}, d\rangle$, which is convex, both $T_1$ and $T_2$ are  monotone.
Maximality follows directly from the maximality of the subdifferential \cite[Theorem 20.25]{bau:17}.
As a consequence, Peaceman-Rachford splitting to \eqref{eq:optconstr} yields the iterates 
\begin{subequations}
\begin{align}
\bar{\lambda}^{(k)} &= J_{cT_1}z^{(k)},\\
z^{(k+1)} &= (R_{cT_2}\circ R_{cT_1})z^{(k)}.
\label{eq:iterates_b}
\end{align}
\label{eq:iterates}
\end{subequations}
\vspace{-\baselineskip}\\
We will first focus on the reflected resolvent $R_{cT_2}$ in (\ref{eq:iterates_b}), which carries the inequality constraints encapsulated by $M$. To do so, we introduce an auxiliary vector $y^{(k)}$, such that 
\begin{align*}
y^{(k)} &= R_{cT_1} z^{(k)}, \\
z^{(k+1)} &= R_{cT_2}y^{(k)}.
\end{align*}
Since $M$ is the intersection of an affine subspace $S = \{\bar{\lambda}\in\bar{\G} \,|\,\bar{\lambda}=P\bar{\lambda} \}$ and a closed convex cone $\bar{K}^\circ$, $M$ is closed and convex, and we have $J_{cT_2}y = {\rm prox}_{c\iota_{M}}(y) =  P_My$, 
the projection of $y$ onto $M$. As a consequence, $R_{cT_2}$ is given by $R_{cT_2} =  2 P_M-I$,  the reflection with respect to $M$, which we will denote by $R_M$. We can explicitly compute $ P_My$, and thus $R_My$.
\begin{lemma} Let $y\in\bar{\G}$. Then
\[
 P_M y = \frac{1}{2}  P_{\bar{K}^\circ} (I+P)y.
\]
\label{lemma:proj}
\end{lemma}
\begin{proof}
We have
\begin{equation}
     P_M y = \arg\min_{u\in M} \| u -y\|^2.
    \label{eq:minu}
\end{equation}
Then $\tilde{u}$ is a solution to \eqref{eq:minu} if and only if
\begin{subequations}
\begin{align}
1. \,\,\,&  \tilde{u} \in \bar{K}^\circ, \, \tilde{u} = P\tilde{u} \label{Pu}\\
2.\,\,\,& (\exists \,\tilde{\xi} \in \bar{K})(\exists\, \tilde{\eta}\in\bar{\G})  \;  \left\{ \begin{array}{l} 0 = 2(\tilde{u}-y) + \tilde{\xi} + (I-P)^*\tilde{\eta}, \\ \langle \tilde{\xi}, \tilde{u}\rangle = 0. \end{array} \right. \label{minL}
\end{align}
\end{subequations}
Combining  \eqref{Pu} and \eqref{minL}, and using the fact that $P^2=I$, we obtain $\tilde{u} = \frac{1}{2}(I+P)y - \frac{1}{4}(I+P) \tilde{\xi}$. Let $v=\frac{1}{2}(I+P)y$ and $w = \frac{1}{4}(I+P) \tilde{\xi}$.
Since $\tilde{\xi}\in \bar{K}$, we have $P\tilde{\xi}\in \bar{K}$ (by construction of $\bar{K}$) and thus $w\in \bar{K}$ since $\bar{K}$ is closed and convex.  Moreover, we have $v =  P_{\bar{K}}v +  P_{\bar{K}^{\circ}}v$ and $w =  P_{\bar{K}}w$ so that $\tilde{u} =  P_{\bar{K}}(v -w) +  P_{\bar{K}^{\circ}}v$. However, since $\tilde{u}\in \bar{K}^\circ$, we conclude that $ P_{\bar{K}}(v -w) = 0$ and thus $\tilde{u}=  P_{\bar{K}^\circ}v$ which completes the proof.
\end{proof}
In order to find a dual expression for $J_{cT_1}z^{(k)}$, 
we note that
\[
{\bar{\lambda}}= J_{cT_1}z \quad \Leftrightarrow \quad z - {\bar{\lambda}} \in cT_1{\bar{\lambda}}.
\]
Hence, ${\bar{\lambda}} = z +c( C{u} - d)$ where ${u} \in \partial f^*(-C^*{\bar{\lambda}})$ {and thus $-C^*{\bar{\lambda}} \in \partial f({u})$. Hence, 
$0 \in \partial f({u}) + C^*{\bar{\lambda}} =  \partial f({u}) + C^*z + cC^*(C{u} - d)$ so that
\[
{u} = \arg\min_x\left(f(x) + \langle z,Cx\rangle + \frac{c}{2}\|Cx-d\|^2\right).
\]
With this, the iterates \eqref{eq:iterates} can  be expressed as
\begin{align}
x^{(k)} &=\displaystyle \arg\min_x\left(f(x) + \langle z^{(k)},Cx\rangle + \frac{c}{2}\|Cx-d\|^2\right), \nonumber \\
\bar{\lambda}^{(k)} &= z^{(k)} + c(Cx^{(k)}-d), \label{eq:pdmmmu} \\
y^{(k)} &= 2\bar{\lambda}^{(k)}  - z^{(k)}, \rule[4mm]{0mm}{0mm} \nonumber \\
z^{(k+1)} &= R_My^{(k)},\rule[4mm]{0mm}{0mm} \nonumber
\end{align}
which can be simplified to
\begin{subequations}
\label{eq:pdmm}
\begin{align}
x^{(k)} &=\displaystyle \arg\min_x\left(f(x) + \langle z^{(k)},Cx\rangle + \frac{c}{2}\|Cx-d\|^2\right),  \label{eq:pdmma}\\
y^{(k)} &= z^{(k)} + 2c(Cx^{(k)}-d),  \label{eq:pdmmb} \\
z^{(k+1)} &= R_My^{(k)}.\rule[4mm]{0mm}{0mm} \label{eq:pdmmc}
\end{align}
\end{subequations}
The iterates \eqref{eq:pdmm} are collectively referred to as the {\em generalised primal-dual method of multipliers} (GPDMM).

Recall that $R_{cT_2} =  2 P_M-I = R_M$.
To get some insight in how to implement $R_M$, note that $R_My =  P_{\bar{K}^\circ} (I+P)y - y$. In the case of equality constraints we have $\bar{K}=\{0\}$ and thus $\bar{K}^\circ = \bar{G}$ so that $R_My = Py$, which is simply a permutation operator.  This permutation operator represents the actual data exchange in the network. That is, we have for all $\{i,j\}\in{\cal E} : z_{i|j} \leftarrow y_{j|i}, \, z_{j|i} \leftarrow y_{i|j}$. 
In the case of general cone constraints, we have $z_{i|j} \leftarrow   P_{{K}^\circ_{ij}}(y_{i|j}+y_{j|i}) - y_{i|j}$ and $z_{j|i} \leftarrow  P_{{K}^\circ_{ij}}(y_{i|j}+y_{j|i}) - y_{j|i}$. 

The distributed nature of PDMM can be made explicit by exploiting the structure of $C$ and $d$
and writing out the update equations \eqref{eq:pdmm}. Recall that $f(x) = \sum_{i\in {\cal V}} \left( f_i(x_i) + \iota_{N_i}(x_i)\right)$. Let
\[
(\forall i\in\V) \;\;\;
C_i:\H_i\to \bar{\G} :  x_i \mapsto \left\{ \begin{array}{ll} A_{ij}x_i, & \text{if } (i,j)\in\bar{\E}, \\ 0, & \text{if } (i,j)\not\in\bar{\E}, \rule[4mm]{0mm}{0mm}\end{array} \right. 
\]
so that $Cx =\sum_{i\in\V} C_ix_i$. In addition,  since $\langle C_j x_j, C_i x_i\rangle = 0$ for $j\neq i$, we have  $\|Cx-d\|^2 = \sum_{i\in\V}\|C_ix_i-d\|^2$.
With this, \eqref{eq:pdmma} for all $i \in \mathcal{V}$ can be expressed as
\begin{align}
 x_i^{(k)} &=\arg\min_{x_i} \left( f_i(x_i)+ \iota_{N_i}(x_i) + \langle z^{(k)}, C_i x_i \rangle +  \frac{c}{2}\|C_i x_i - d\|^2 \right) \nonumber \\
&=\arg\min_{x_i\in N_i}\left( f_i(x_i) + \langle z^{(k)}, C_i x_i \rangle +  \frac{c}{2}\|C_i x_i - d\|^2 \right) \label{eq:xic}  \\
&=\arg\min_{x_i\in N_i} \bigg(f_i(x_i) \hspace{-0.7mm}+\hspace{-0.7mm}  \sum_{j\in\mathcal{N}_i}\!\bigg(\langle z_{i|j}^{(k)},A_{ij}x_i\rangle \hspace{-0.7mm}+\hspace{-0.7mm}  \frac{c}{2}\|A_{ij}x_i \hspace{-0.7mm} - \hspace{-0.7mm} \frac{1}{2}b_{ij}\|^2 \bigg) \!\! \bigg). \nonumber
\end{align}
In addition, we can express \eqref{eq:pdmmb} as
\[
(\forall i\in\V) (\forall j \in {\cal N}_i) \,\,\, y_{i|j}^{(k)}  = z_{i|j}^{(k)} + 2c\left(A_{ij}x_i^{(k)}  - \frac{1}{2}b_{ij} \right),
\]
after which data is exchanged amongst neighbouring nodes, and the auxiliary variables are updated:
\begin{equation}
    (\forall i\in\V) (\forall j \in {\cal N}_i) \,\,\, z_{i|j}^{(k+1)} =    P_{{K}^\circ_{ij}}\left(y_{i|j}^{(k)}+y_{j|i}^{(k)}\right) - y_{i|j}^{(k)}.
\label{eq:zupdate}
\end{equation}
The resulting algorithm is visualised in the pseudo-code of Algorithm~\ref{alg:pdmm}. 
It can be seen that no direct collaboration is required between nodes during the computation of these updates, leading to an attractive (parallel) algorithm for optimisation in practical networks. The update \eqref{eq:pdmmc} can be interpreted as one-way transmissions of the auxiliary $y$ variables to neighbouring nodes where the actual update of the $z$ variables is done.
\begin{algorithm}[t]
\caption{Synchronous GPDMM.}\label{alg:PDMMuni}
\begin{algorithmic}[1]
\State\textbf{Initialise:}$\quad z^{(0)}\in\bar{\G}, c \in \R_{++}$ \Comment{Initialisation}
\For{$k=0,...,$}
    \For{$i\in \cal V$}\Comment{Node updates}
        \State $\displaystyle x_i^{(k)} =\arg\min_{x_i\in N_i}\bigg( f_i(x_i)+ \sum_{j\in\mathcal{N}_i}\!\left(\langle z_{i|j}^{(k)},A_{ij}x_i\rangle +  \frac{c}{2}\|A_{ij}x_i - \frac{1}{2}b_{ij}\|^2  \right)\!\!\bigg)$
        \ForAll{$j\in\mathcal{N}_i$}
            \State $y_{i|j}^{(k)} = z_{i|j}^{(k)} + 2c \left(A_{ij}x_i^{(k)} - \frac{1}{2}b_{ij}\right)$
        \EndFor
    \EndFor
    \item[]
    \ForAll{$i\in{\cal V}, j\in{\cal N}_i$}\Comment{Transmit variables}
        \State $\textbf{node}_j\leftarrow\textbf{node}_i\left(y_{i|j}^{(k)}\right)$
    \EndFor
    \item[]
    \ForAll{$i\in{\cal V}, j\in{\cal N}_i$}\Comment{Auxiliary updates}
        \State $z_{i|j}^{(k+1)} =  P_{{K}^\circ_{ij}}\left(y_{i|j}^{(k)}+y_{j|i}^{(k)}\right) - y_{i|j}^{(k)}$
    \EndFor
\EndFor
\end{algorithmic}
\label{alg:pdmm}
\end{algorithm}

\subsection{Node constraints}
\label{sec:node}

So far we considered constraints of the form $A_{ij}x_i + A_{ji}x_j  - b_{ij}\in K_{ij}$. If we set $A_{ji}$ to be the zero operator, we have constraints of the form $A_{ij}x_i - b_{ij}\in K_{ij}$, which are node constraints; it sets constraints on the values  $x_i$ can take on. 
Even though $x_j$ is not involved in the constraint anymore, there is still communication needed between node $i$ and node $j$ since at the formulation of the lifted dual problem \eqref{eq:exdual} we have introduced two auxiliary variables, $\lambda_{i|j}$ and $\lambda_{j|i}$, one at each node, to control the constraints between node $i$ and $j$. This was done independent of the actual value of $A_{ij}$ and $A_{ji}$. In order to guarantee convergence of the algorithm, these variables need to be updated and exchanged during the iterations. 
To avoid such communication between nodes, we can introduce dummy nodes, one for every node that has a node constraint. Let $i'$ denote the dummy node introduced to define the node constraint on node $i$. That is, we have $A_{ii'}x_i - b_{ii'}\in K_{ii'}$. Since dummy node $i'$ is only used to communicate with node $i$, it is a fictive node and can be incorporated in node $i$, thereby avoiding any network communication for node constraints. In such cases, we will simply write  $A_{i}x_i - b_{i}\in K_{i}$.

\begin{example}
Let $(\forall i\in\V) \,\, \H_i= S^n$, and consider the node constraints $(\forall i\in\V) \, X_i \in S^n_+$, the cone of positive semidefinite matrices in $\R^{n\times n}$. Hence $(\forall i\in\V) \, K_i = S^n_+$. Since $K_i^\circ = - K_i^* = -S^n_+=S^n_-$, where $S^n_-$ is the cone of negative semidefinite matrices, so that \eqref{eq:zupdate} becomes $(\forall i\in\V)(\forall j\in{\cal N}_i) \, Z_{i|j} =  P_{S^n_-}\left(Y_{i|j}+Y_{j|i}\right) - Y_{i|j}$, where
$ P_{S^n_-}Y = Q\Lambda_- Q^T$, where $Q\Lambda Q^T$ is the eigenvalue decomposition of $Y$ and  $\Lambda_-$ is the matrix obtained from $\Lambda$ by setting the positive entries to $0$.  
\label{ex:posdef}
\end{example}

\section{Local inequality constraints}
\label{sec:ineq}

By inspection of \eqref{eq:xic}, we observe that \eqref{eq:xic} is a constraint optimisation problem due to the fact that we have included the primal constraints $(\forall{\ell_i} \in{\cal L}_i) \,\, h_{\ell_i} (x_i)\leq 0$ in the objective function  using the indicator function $\iota_{N_i}$.
For many practical problems, however, the constraint optimisation problem can be solved analytically and efficiently, as the following example shows. 
\begin{example}
Let $(\forall i\in\V) \,\,\H_i = S^n$ equipped with the scalar product $(X,Y) \mapsto {\rm tra}(X^TY)$, and consider the following consensus problem: 
\begin{equation}
\begin{array}{lll} \text{minimise} & \displaystyle \sum_{i\in\V} \frac{1}{2}\|X_i-Q_i\|^2 \\
\text{subject to} &(\forall i\in\V) &\hspace{-8mm}\left\{  \!\!\begin{array}{l} X_i  \in K_i, \\ (\forall {\ell_i} \in {\cal L}_i) \,\,\,{\rm tra}(H_{\ell_i} ^TX_i) = 1,\rule[4mm]{0mm}{0mm} \end{array}\right. \\
&(\forall \{i,j\}\in\E) &\hspace{-8mm}X_i=X_j.\rule[4mm]{0mm}{0mm}
\end{array}
\label{eq:Qx}
\end{equation}
By inspection of \eqref{eq:Qx} we note that the constraints $(\forall i\in\V) \, X_i \in K_i$ and $(\forall \{i,j\}\in\E) \, X_i=X_j$ will be handled by the GPDMM iterates, while the constraints $(\forall{\ell_i} \in{\cal L}_i) \,\,{\rm tra}(H_{\ell_i} ^TX_i) = 1$ appear as constraints in the update \eqref{eq:xic}. Since the objective function is quadratic, \eqref{eq:xic} is the minimisation over a quadratic function subject to affine constraints, which boils down to solving a set of linear equations. Indeed, note  
that the equality constraints over all edges imply that $(\forall \{i,j\}\in\E) \,\,A_{ij} = -A_{ji} = I, b_{ij}=0$, and thus that $C_i^*C_j = ({\rm deg}\,i)\delta_{ij}$, where $\delta_{ij}$ is the Kronekcer delta. The solution $X_i^{(k)}$ to \eqref{eq:xic} for problem \eqref{eq:Qx} therefore satisfies the following optimality conditions:
\begin{align}
1. \,\,\,&  (\forall i\in\V) (\forall {\ell_i} \in{\cal L}_i) \,\,\, {\rm tra}(H_{\ell_i} ^TX^{(k)}_i)  = 1, \label{eq:La}\\
2.\,\,\,&  (\forall i\in\V)  (\exists \,\gamma^{(k)}_i \in\R^{n_i})  \;  \nonumber\\
& (1+c\,{\rm deg}\,i) X_i^{(k)} - Q_i + C_i^*Z^{(k)}  + \sum_{{\ell_i} \in{\cal L}_i} (\gamma^{(k)}_i)_{\ell_i}  H_{\ell_i}  = 0. \nonumber
\end{align}
Hence, 
\begin{align*}
(\forall i\in\V) \,\,\, X_i^{(k)} &=  \left(Q_i - C_i^*Z^{(k)} -\sum_{{\ell_i} \in{\cal L}_i} (\gamma^{(k)}_i)_{\ell_i}   H_{\ell_i}  \right)/(1 +c\,{\rm deg}\,i) \\
&= \tilde{X}_i^{(k)} - \sum_{{\ell_i} \in{\cal L}_i} (\tilde{\gamma}^{(k)}_i)_{\ell_i}  H_{\ell_i} , 
\end{align*}
where $\tilde{X}_i^{(k)}$ is the solution of \eqref{eq:xic} in the absence of the constraints, and $\tilde{\gamma}^{(k)}_i = \gamma^{(k)}_i/(1+c\,{\rm deg}\,i)$.
The Lagrange multipliers $\tilde{\gamma}^{(k)}_i$ are found by solving \eqref{eq:La}. Let $(\forall i\in\V) \, G_i \in\R^{n_i\times n_i}$ having entries $(\forall m\in{\cal L}_i)(\forall n\in{\cal L}_i) \, G_{i_{m,n}}  = {\rm tra}(H^T_{\ell_m}H_{\ell_n})$. With this,  \eqref{eq:La} for all $i\in \mathcal{V}$ can be expressed as
\begin{equation}
\left( \begin{array}{c} 1\\\vdots\\1\end{array} \right) = 
\left( \!\begin{array}{c}  {\rm tra}(H_1^T{X}^{(k)}_i) \\ \vdots \\{\rm tra}(H_{n_i}^T{X}^{(k)}_i) \end{array}\! \right) = 
 \left( \!\begin{array}{c}  {\rm tra}(H_1^T\tilde{X}^{(k)}_i) \\ \vdots \\{\rm tra}(H_{n_i}^T\tilde{X}^{(k)}_i) \end{array}\! \right) -
 G_i \tilde{\gamma}^{(k)}_i.
   \label{eq:gamma}
\end{equation}
In conclusion, the constraints $(\forall {\ell_i} \in{\cal L}_i) \,\,{\rm tra}(H_{\ell_i} ^TX_i) = 1$ can be implemented by applying a simple correction factor $\sum_{{\ell_i} \in{\cal L}_i} (\tilde{\gamma}^{(k)}_i)_{\ell_i}  H_{\ell_i} $ to $\tilde{X}_i^{(k)}$, the solution of \eqref{eq:xic} in the absence of the constraints. A similar conclusion holds for semidefinite programs where the objective function is linear, and \eqref{eq:xic} is still a minimisation over a quadratic function subject to affine constraints. For more complex objective functions the inequality constraints can be solved locally using standard convex solvers, or we can replace the complex objective function at every iteration by a quadratic approximation, leading to simple analytic update equations as described above, a procedure that is also guaranteed to converge \cite{con:18b}.
\end{example}

\section{Convergence of the GPDMM algorithm}
\label{sec:convergence}

Let $T = R_{cT_2}\circ R_{cT_1}$. Since both $R_{cT_2}$ and $R_{cT_1}$ are nonexpansive, $T$ is nonexpansive, and  the sequence generated by the Banach–Picard iteration $z^{(k+1)} = Tz^{(k)}$ may fail to produce a fixed point of $T$. A simple example of this situation is $T=-I$ and $z^{(0)} \neq 0$.
Although operator averaging provides a means of ensuring algorithmic convergence, resulting in the Krasnosel'skii-Mann iterations $z^{(k+1)} = (1-\alpha)z^{(k)} + \alpha Tz^{(k)}$ with $\alpha\in (0,1]$, it is well known that Banach-Picard iterations converge provable faster than Krasnosel'skii-Mann iterations for the important class of quasi-contractive operators \cite{ber:04}. 
As discussed before, the Peaceman-Rachford splitting algorithm converges when $T_1$ is uniformly monotone. 
However, assuming finite dimensional Hilbert spaces, we have ${\rm dim}\,\bar{\G} > {\rm dim}\,\H$ for any practical network so that ${\rm ker}\,C^* \neq\emptyset$. 
Recall that $T_1 = -C \partial f^*(-C^*(\cdot)) + d$.
Hence, $(\exists (\bar{\lambda},\eta)\in \bar{\G}\times \bar{\G}) \; \bar{\lambda}\neq \eta \text{ and } C^*(\bar{\lambda}-\eta) = 0$, which prohibits $T_1$ of being strictly monotone, and thus uniformly monotone. 
It is therefore of interest to consider if there are milder conditions under which 
certain optimality can be guaranteed. Whilst such conditions may be restrictive in the case of convergence of the auxiliary variables, in the context of distributed optimisation we are often only interested in primal optimality. For this reason we define conditions that ensure $x^{(k)}\to x^*$ even if $z^{(k)} \not\to z^*, z^*\in {\rm fix}\,T$. 

\begin{proposition}
Let $T_1 = -C \partial f^*(-C^*(\cdot)) + d$ and $T_2 = \partial \iota_{M}$ such that ${\rm fix}\,T \neq \emptyset$ and $\partial f$ is uniformly monotone with modulus $\phi$, let $c\in\R_{++}$, and let $x^*$ be the solution to the primal problem \eqref{eq:primal}. Given the iterates \eqref{eq:iterates} and $z^{(0)} \in \bar{\G}$, we have $x^{(k)} \to x^*$.
\label{prop:conv}
\end{proposition}
\begin{proof}
Let $z^* \in {\rm fix}\,T$. We have for all $k\in\N$,
\begin{align}
&\|z^{(k+1)} - z^*\|^2 \nonumber \\
&= \|(R_{cT_2}\circ R_{cT_1})z^{(k)}  - (R_{cT_2}\circ R_{cT_1})z^*\|^2 \nonumber \\
&\leq \|R_{cT_1} z^{(k)} -  R_{cT_1} z^*\|^2  \nonumber \\
&= \|2\bar{\lambda}^{(k)} - z^{(k)} - (2\bar{\lambda}^* - z^*)\|^2 \nonumber \\
&= \|z^{(k)}-z^*\|^2 -4\,\langle \bar{\lambda}^{(k)} -\bar{\lambda}^*, z^{(k)}-\bar{\lambda}^{(k)} - (z^*-\bar{\lambda}^*)\rangle  \nonumber \\
&=  \|z^{(k)}-z^*\|^2+ 4c\,\langle \bar{\lambda}^{(k)} -\bar{\lambda}^*, C(x^{(k)}-x^*)\rangle,
\label{eq:zCT}
\end{align}
where the last equality follows from \eqref{eq:pdmmmu}. Moreover, since $x^{(k)}$ minimises $f(x) + \langle z^{(k)},Cx\rangle + \frac{c}{2}\|Cx-d\|^2$, we have that
$0\in \partial f(x^{(k)}) + C^*z^{(k)} + cC^*(Cx^{(k)}-d) =  \partial f(x^{(k)}) + C^*\bar{\lambda}^{(k)}$, so that
\eqref{eq:zCT} can be expressed as
\begin{align}
\|z^{(k+1)} - z^*\|^2 &\leq  \|z^{(k)}-z^*\|^2- 4c\,\langle \partial f(x^{(k)}) -\partial f(x^*), x^{(k)}-x^*\rangle  \nonumber \\
&\leq  \|z^{(k)}-z^*\|^2 - 4c \phi(\|x^{(k)}-x^*\|).
\label{eq:monof}
\end{align}
Hence, $\phi(\|x^{(k)}-x^*\|)\to 0$ and, in turn, $\|x^{(k)}-x^*\|\to 0$. 

To show that $x^*$ is primal feasible, we consider two successive 
$z$-updates:
\begin{align}
z^{(k+1)} &= R_M\left(z^{(k)} + 2c(Cx^{(k)}-d)\right) \nonumber \\
&= z^{(k-1)} + 2c\left((Cx^{(k-1)} -d)+ R_M( Cx^{(k)} - d)\right).
\label{eq:zzg}
\end{align}
Subtracting $z^*$ from both sides of \eqref{eq:zzg} and observing from \eqref{eq:monof} that $(\|z^{(k)} - z^*\|^2)_{k\in\N}$ converges, we conclude that $0 = Cx^{*} -d + R_M( Cx^{*} - d) = 2 P_M(Cx^{*} - d ) =  P_{\bar{K}^\circ} (I+P)(Cx^*-d)$. Hence, $(I+P)(Cx^*-d) \in \bar{K}$ and thus $Ax^*-b\in K$ by \eqref{eq:CxAx}, which completes the proof.
\end{proof}
\begin{remark}
Since $T$ is at best nonexpansive, the auxiliary variables will not converge in general. In fact, from \eqref{eq:zzg} we conclude that 
$(z^{(2k)})_{k\in\N}$ and $(z^{(2k+1)})_{k\in\N}$ converges, and similarly for $y$ and $\bar{\lambda}$. 
\end{remark}

\section{Stochastic coordinate descent}
\label{sec:stoch}

In practice, synchronous algorithm operation necessitates a global clocking system to coordinate actions among nodes. Clock synchronisation, however, in particular in large-scale heterogeneous sensor networks, can be challenging. In addition, in such heterogeneous environments where sensors or agents vary greatly in processing capabilities, processors that are fast either because of high computing power or because of small workload per iteration are often forced to idle while waiting for slower processors to catch up. Asynchronous algorithms offer a solution to these issues by providing greater flexibility in leveraging information received from other processors, thereby mitigating the constraints imposed by synchronous operation.
In order to obtain an asynchronous (averaged) GPDMM algorithm, we will apply randomised coordinate descent to the algorithm presented in Section~\ref{sec:operator}. 

Stochastic updates can be defined by assuming that each auxiliary variable $z_{i|j}$ can be updated based on a Bernoulli random variable $\chi_{i|j}\in\{0,1\}$. Let $\chi = (\chi_{i|j})_{(i,j)\in\bar{\E}}$, and let
$(\chi^{(k)})_{k\in\mathbb{N}}$ be an i.i.d.\ random process defined on a common probability space $(\Omega, {\cal A},{\cal P})$, such that $\chi^{(k)} : (\Omega,{\cal A}) \to \{0,1\}^{|\bar{\E}|}$. 
Hence, $\chi^{(k)}(\omega) \subseteq \{0,1\}^{|\bar{\E}|}$ indicates  which entries of $z^{(k)}$ will be updated at iteration $k$. We assume that the following condition holds:
\begin{equation}
    (\forall (i,j) \in \bar{\E}) \quad {\cal P}(\{\chi_{i|j}^{(0)} = 1\})>0.
\label{eq:xiex}
\end{equation}
Since $(\chi^{(k)})_{k\in\mathbb{N}}$ is i.i.d., \eqref{eq:xiex} guarantees that at every iteration, entry $z_{i|j}^{(k)}$ has nonzero probability to be updated. We define the stochastic operator $U: \bar{\G}\to \bar{\G} : (z_{i|j})_{(i,j)\in\bar{\E}} \mapsto (\chi_{i|j}z_{i|j})_{(i,j)\in\bar{\E}}$.
With this,  we define the {\em stochastic} Banach-Picard iteration \cite{jor:23} as
\begin{equation}
    Z^{(k+1)} =  T_{\chi^{(k)}} Z^{(k)} = \big(I-U^{(k)}\big)Z^{(k)}+U^{(k)}(TZ^{(k)}),
\label{eq:sbp}  
\end{equation}
where $Z^{(k)}$ denotes the random variable having realisation $z^{(k)}$.
If $T$ is $\alpha$-averaged, a convergence proof is given in \cite{lut:13,bia:16}, where it is shown that $Z^{(k)}  - T_{\chi^{(k)}} Z^{(k)} \stackrel{\rm a.s.}{\to} 0$  (almost surely).
If $T$ is not averaged, we do not have convergence in general since $T$ is at best nonexpansive and we need additional conditions for convergence.  We have the following convergence result for stochastic GPDMM.
\begin{proposition}
Let $T_1 = -C \partial f^*(-C^*(\cdot)) + d$ and $T_2 = \partial \iota_{M}$ such that ${\rm fix}\,T \neq \emptyset$ and $\partial f$ is uniformly monotone with modulus $\phi$, let $c\in\R_{++}$, and let $x^*$ be the solution to the primal problem \eqref{eq:primal}. Given the stochastic iteration \eqref{eq:sbp} and $z^{(0)} \in \bar{G}$, we have $X^{(k)} \stackrel{\rm a.s.}{\to} x^*$.    
\label{prop:sconv}
\end{proposition}
\begin{proof}
Let $\bar{\E}_\chi =  \{(i,j) \in\bar{\E} \,|\, \chi_{i|j}=1\}$ and let  $p_{i|j} = {\cal P}(\{\chi_{i|j}^{(0)} = 1\})$. Moreover, let $S = \{0,1\}^{|\bar{\E}|}$ so that $(\forall \chi\in S) \; p_\chi = \sum_{(i,j)\in \bar{\E}_\chi} p_{i|j}$.  We define $\langle \cdot\,,\, \cdot\rangle_{\cal P}$ as 
\[
(\forall u\in\bar{\G})(\forall v\in\bar{\G}) \quad \langle u,v\rangle_{\cal P} = \!\!\!\sum_{(i,j)\in\bar{\E}} p_{i|j}^{-1} \langle u_{i|j},v_{i|j} \rangle.
\]
In addition, let  ${\cal A}_k := \sigma\{\chi^{(t)}: t\leq k\}$, the $\sigma$-algebra generated by the random vectors $\chi^{(0)},\ldots,\chi^{(k)}$. Since  ${\cal A}_k \subseteq {\cal A}_l$ for $k\leq l$, the sequence $({\cal A}_k)_{k\in\N}$ is a filtration on $(\Omega,{\cal A})$.
For any $z^* \in {\rm fix}\,T$ we have (see also \cite{lut:13,bia:16})
\begin{align}
\mathbb{E} &\left( \|Z^{(k+1)} - z^*\|_{{\cal P}}^2 \, | \, {\cal A}_k \right)  =  \sum_{\chi\in S} p_\chi \|T_{\chi}Z^{(k)}-z^*\|^2_{\cal P} \nonumber \\
&\hspace{0mm}=  \sum_{\chi\in S} p_\chi  \bigg( \sum_{(i,j)\in\bar{\E}_\chi}  \|(T Z^{(k)})_{i|j}-z_{i|j}^*\|^2_{\cal P}  + \sum_{(i,j)\not\in\bar{\E}_\chi}  \|Z_{i|j}^{(k)}-z_{i|j}^*\|^2_{\cal P} \bigg) \nonumber \\
&\hspace{0mm}=   \|Z_{i|j}^{(k+1)}-z_{i|j}^*\|_{{\cal P}}^2 \nonumber \\
&\hspace{3mm}+ \sum_{\chi\in S} p_\chi  \sum_{(i,j)\in\bar{\E}_\chi}p_{i|j}^{-1}  \left(   \|(T Z^{(k)})_{i|j}-z_{i|j}^*\|^2 -  \|Z_{i|j}^{(k)}-z_{i|j}^*\|^2 \right)  \nonumber \\
&\hspace{0mm}=  \|Z_{i|j}^{(k+1)}-z_{i|j}^*\|_{{\cal P}}^2 +  \sum_{(i,j)\in\bar{\E}} \left(  \|(T Z^{(k)})_{i|j}-z_{i|j}^*\|^2 -  \|Z_{i|j}^{(k)}-z_{i|j}^*\|^2 \right)  \nonumber \\
&\hspace{0mm}= \|Z^{(k)} - z^*\|_{{\cal P}}^2 +  \| TZ^{(k)} - z^*\|^2 - \|Z^{(k)} - z^*\|^2.
\label{eq:Eless}
\end{align}
Using \eqref{eq:monof}, \eqref{eq:Eless} becomes
\begin{align}
\mathbb{E} \left( \|Z^{(k+1)} - z^*\|_{{\cal P}}^2\right. & \left. \!\!| \, {\cal A}_k \right) \leq \|Z^{(k)} - z^*\|_{{\cal P}}^2 - 4c\phi(\|X^{(k)}-x^*\|),
\label{eq:mart}
\end{align}
which shows that $(\|Z^{(k)} - z^*\|_{{\cal P}}^2)_{k\geq 1}$ is a nonnegative supermartingale. Moreover, since  $(\,\cdot\,)^{1/2}$ is concave and nondecreasing on $\R_+$, we conclude 
that $(\|Z^{(k)} - z^*\|_{{\cal P}})_{k\geq 1}$ is a nonnegative supermartingale as well and therefore converges almost surely by the martingale convergence theorem \cite[Theorem 27.1]{jac:04}. 
Taking expectations on both sides of \eqref{eq:mart} and iterating over $k$, 
we obtain 
\[
    \mathbb{E}\left(\|Z^{(k+1)} - z^*\|_{{\cal P}}^2\right) \leq \|z^{(0)} - z^*\|_{{\cal P}}^2 - 4c\sum_{t=1}^{k}\mathbb{E}\left(\phi(\|X^{(t)}-x^*\|)\right),
\]
so that
\begin{equation*}
    \sum_{t=1}^{k}\mathbb{E}\left(\phi(\|X^{(t)}-x^*\|)\right)\leq\frac{1}{4c}\|z^{(0)} - z^*\|_{{\cal P}}^2<\infty,
\end{equation*}
which shows that the sum of the expected values of the primal error is bounded. Hence, using Markov's inequality, we conclude that
\[
  \sum_{t=1}^\infty {\rm Pr}\left\{\|X^{(t)}-x^*\|^2\geq\epsilon\right\} \leq\frac{1}{\epsilon} \sum_{t=1}^\infty\mathbb{E}\left[\|X^{(t)}-x^*\|^2\right]<\infty,
\]
for all $\epsilon>0$,
and by Borel Cantelli's lemma \cite[Theorem 10.5]{jac:04}  that
\[
    {\rm Pr}\left\{\limsup_{k\rightarrow\infty}\left(\|X^{(k)}-x^*\|^2\geq\epsilon\right)\right\}=0,
\]
which shows that $\|X^{(k)}-x^*\|^2\stackrel{\rm a.s.}{\to} 0$.  The proof that $x^*$ is primal feasible is identical to the one given in the proof of Theorem~\ref{prop:conv}.
\end{proof}

\section{Applications and numerical experiments}
\label{sec:num}

In this section we will discuss experimental results obtained by Monte Carlo simulations. We will start in Section~\ref{sec:cons} by showing convergence results for both synchronous and asynchronous updating schemes, where we optimise a quadratic objective function subject to different cone constraints. Next, in Section~\ref{sec:mimo}, we will show simulation results for a practical communication application of decentralised multiple-input multiple-output (MIMO) maximum likelihood (ML) detection is an ad-hoc communication network. The problem turns out to be NP-hard but can be well approximated through semidefinite relaxation techniques.

\subsection{Cone constrained consensus problem}
\label{sec:cons}

\begin{figure}[t]
\centering
\includegraphics[width=.6\textwidth]{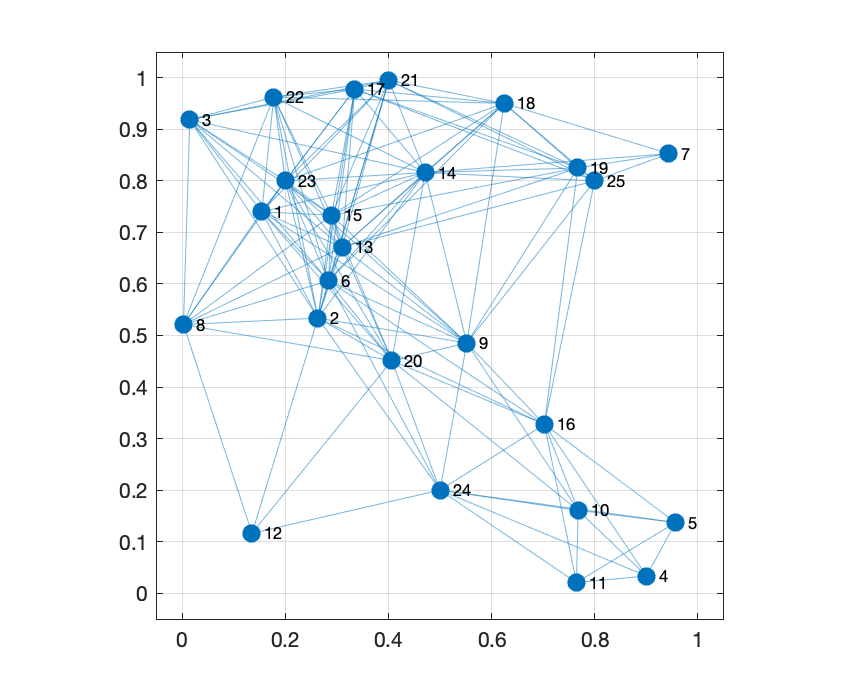}
\vspace{-.7\baselineskip}
\caption{Demonstration of a random geometric graph with 25 nodes.} 
\label{fig:graph_25nodes}
\end{figure}
In our first computer simulation
we consider a random geometric graph of $N=25$ nodes where we have set the communication radius $r = \sqrt{2\log(N)/N}$, thereby guaranteeing a connected graph with probability at least $1-1/N^2$ \cite{dal:02}.
The resulting graph is depicted in Fig.~\ref{fig:graph_25nodes}.

Let $m$ and $n$ be strictly positive 
integers. Let $(\forall i\in\V) \, \H_i = \R^{m\times n}$ equipped with the scalar product $(X,Y) \mapsto {\rm tra}(X^TY)$. The associated norm is the Frobenius norm $\|\cdot\|_{\rm F}$. 
We consider the following consensus problem:
\begin{equation}
\begin{array}{lll} \text{minimise} & \displaystyle\sum_{i\in\V}\|X_i-A_i\|^2_{\rm F} \\ 
\text{subject to} &(\forall \{i,j\}\in {\cal E}) &\hspace{-6mm}X_i=X_j,
\rule[4mm]{0mm}{0mm} \\
 &(\forall i \in\V) &\hspace{-6mm}X_i \in K_i,
\rule[4mm]{0mm}{0mm} 
\end{array}
\label{eq:fro}
\end{equation}
where the data $A_i\in\H_i$ was randomly generated from a zero mean, unit variance Gaussian distribution. Problem \eqref{eq:fro} is in standard form and can be solved directly with the GPDMM algorithm. 
We will consider two scenarios.  Let $\bar{A} = |\V|^{-1}\sum_{i\in\V}A_i$. The first scenario is the one in which we set 
$(\forall i\in\V) \, K_i = \R_+^{m\times n}$, the cone of real nonnegative matrices, and the  solution to \eqref{eq:fro} is given by $(\forall i\in\V) \, X_i^* = \bar{A}_+$, where $A_+$ is the matrix obtained from $A$ by setting the negative entries to $0$. Fig.~\ref{fig:fro}a) shows convergence results for $m=5$ and $n=10$.
Results are shown for both synchronous (solid lines) and asynchronous (dashed lines) update schemes, and for $\alpha \in\{\frac{1}{2},1\}$. The case $\alpha=1$ corresponds to the non-averaged Picard-Banach iterations (Peaceman-Rachford splitting algorithm), while $\alpha=\frac{1}{2}$ corresponds to the Krasnosel'skii-Mann iterations (Douglas-Racheford splitting algorithm). 
\begin{figure}[t]
\centering
\includegraphics[width=.49\textwidth]{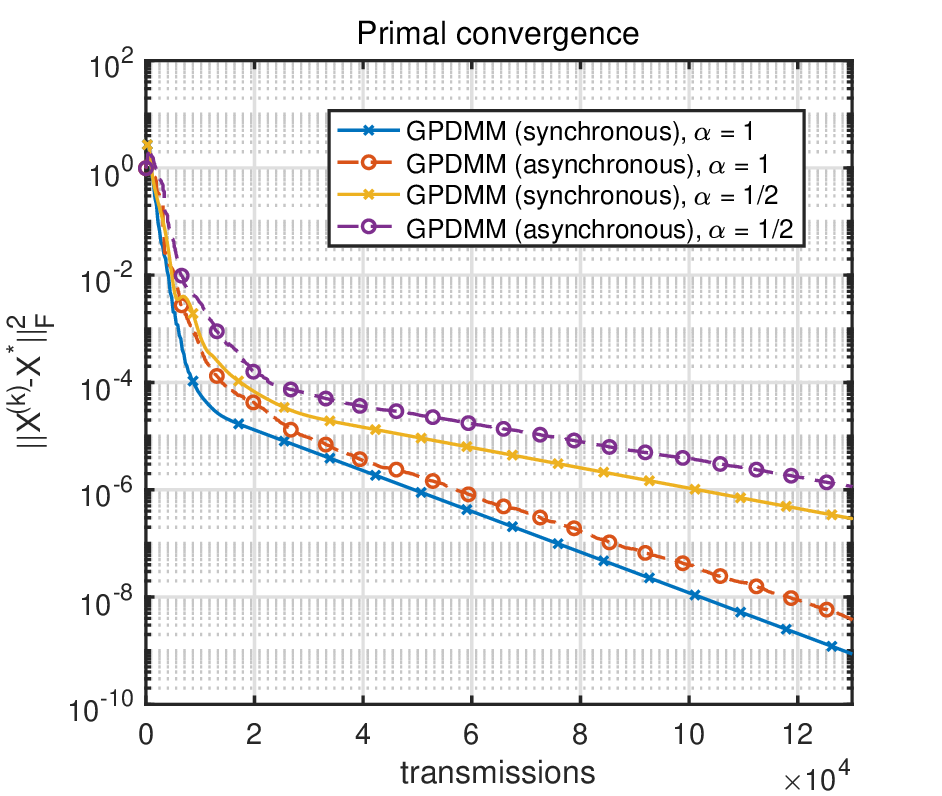}
\includegraphics[width=.49\textwidth]{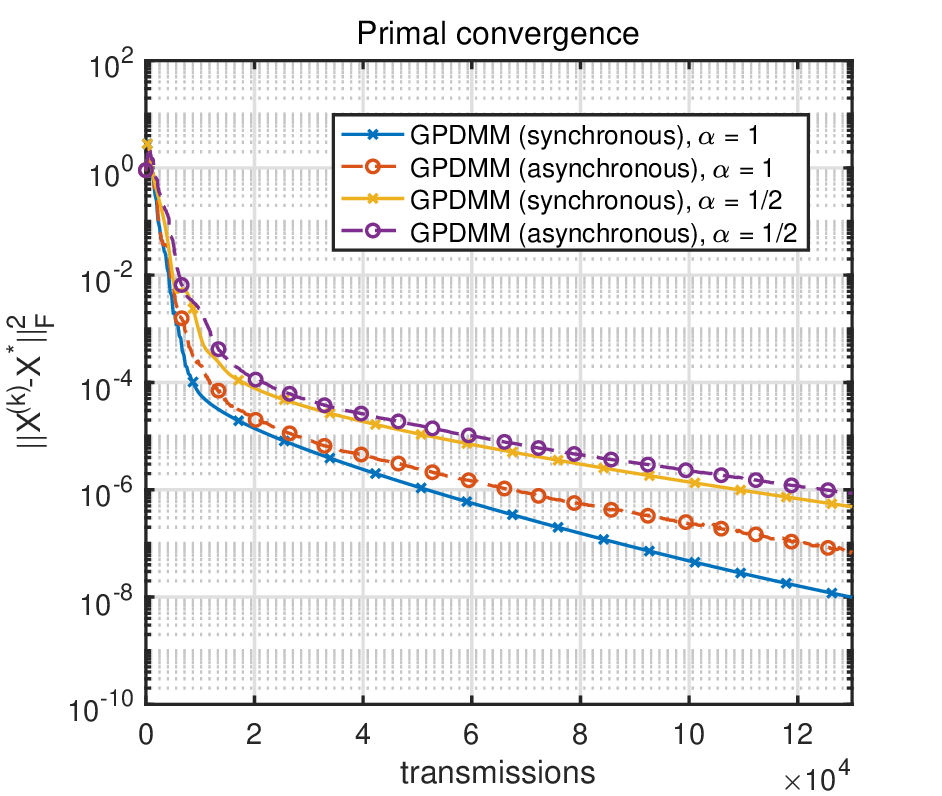}\\
\hspace{3mm} $a) \hspace{63mm} b)$ \vspace{-.3\baselineskip}
\caption{Convergence results for GPDMM for the consensus problem \eqref{eq:fro} over the graph depicted in Fig.~\ref{fig:graph_25nodes} for (a) $(\forall i\in\V) \, K_i = \R_+^{5\times 10}$ and (b) $(\forall i\in\V) \, K_i = S^{10}_+ $.} 
\label{fig:fro}
\end{figure}
In the second scenario, we 
 consider the case where $(\forall i\in\V) \, \H_i = S^{n}$, again equipped with the scalar product $(X,Y) \mapsto {\rm tra}(X^TY)$ and 
$(\forall i\in\V) \, K_i = S^{n}_+$, and the  solution to \eqref{eq:fro} is given by $(\forall i\in\V) \, X_i^* = Q\Lambda_+ Q^T$, where $Q\Lambda Q^T$ is the eigenvalue decomposition of $\bar{A}$. Fig.~\ref{fig:fro}b) shows convergence results for  $n=10$. The figures show that synchronous and asynchronous implementations have similar convergence results and that, as expected, averaging slows down the convergence rate.

\subsection{Decentralised approximation algorithms for max-cut problems}
\label{sec:maxcut}


Given an undirected graph $G = (\V,\E)$, we define nonnegative weights $(\forall \{i,j\}\in\E) \, w_{ij} = w_{ji}$ on the edges in the graph, and we set $(\forall \{i,j\} \not\in\E) \, w_{ij} = w_{ji} = 0$ when there is no edge between nodes. The maximum cut (max-cut) problem is that of finding the set of vertices $S \subseteq \V$ that maximises the weight of the edges in the cut $(S,\bar{S})$, where $\bar{S} = \V\backslash S$. The maximum cut problem is known to be NP-hard and a typical approach to solving such a problem is to find a $p$-approximation algorithm which is a polynomial-time algorithm that delivers a solution of value at least $p$ times the optimal value. Goemans and Williamson \cite{goe:95} gave the $0.878$-approximation algorithm  based on semidefinite programming and randomised rounding. This is the best known approximation guarantee for the max-cut problem today.

The max-cut problem is given by the following integer quadratic program
\begin{equation}
\begin{array}{ll} \text{maximise} & \displaystyle \frac{1}{4} \sum_{i\in\V} \sum_{j\in{\cal N}_i} w_{ij} ( 1 - x_ix_j) \\
\text{subject to} & (\forall i\in\V) \,\, x_i \in  \{-1,1\}.\rule[4mm]{0mm}{0mm}
\end{array}
\label{eq:maxcut}
\end{equation}

We can find an approximate solution of \eqref{eq:maxcut} by reformulating the problem as a homogeneous (nonconvex) QCQP and solve a semidefinite relaxed version of it. To do so, let $n = |\V|$, and let $x  = (x_1,\ldots, x_n)^T$. Ignoring the scaling, \eqref{eq:maxcut} can be equivalently expressed as
\begin{equation}
\begin{array}{ll} \text{maximise} & x^TW x \\
\text{subject to} & x \in  \{-1,1\}^{n},\rule[4mm]{0mm}{0mm}
\end{array}
\label{eq:mcQCQP}
\end{equation}
where $W\in \R^{n\times n}$ given by
\[
(\forall i\in\V)(\forall j\in\V) \,\,\,\,W_{ij} =  \left\{ \begin{array}{ll}\displaystyle \sum_{j\in{\cal N}_i} w_{ij}, & \text{if }  i=j,\\ -w_{ij}, & \text{if }  i\neq j, \{i,j\}\in\E, \rule[4mm]{0mm}{0mm} \\
0, & \text{if } i\neq j, \{i,j\}\not\in\E. \rule[4mm]{0mm}{0mm}\end{array} \right.
\]
Note that $W\succeq 0$.
Let ${\cal L} = \{1,\ldots,n\}$. 
Since $x^TWx = {\rm tra}(x^TWx) =  {\rm tra}(Wxx^T)$, we can rewrite \eqref{eq:mcQCQP} as
\begin{equation}
\begin{array}{ll} \text{maximise} & \displaystyle  {\rm tra}(WX) \\
\text{subject to} & X \succeq 0, \rule[3.5mm]{0mm}{0mm}\\
& {\rm rank}(X) = 1, \rule[3.5mm]{0mm}{0mm}\\
&(\forall \ell\in{\cal L}) \,\,\, X_{\ell\ell} = 1.\rule[3.5mm]{0mm}{0mm}
\end{array}
\label{eq:mctra}
\end{equation}
The condition $X\succeq 0$ ($X$ positive semidefinite) in combination with ${\rm rank}(X)=1$ implies $X = xx^T$. The usefulness of expressing the original max-cut problem \eqref{eq:maxcut} in the form \eqref{eq:mctra} lies in the fact that \eqref{eq:mctra} allows us to identify the fundamental difficulty in solving
\eqref{eq:maxcut}, which is the nonconvex rank constraint; the objective function as well as the other constraints are convex.t Hence, we can directly relax this problem into a convex problem by ignoring the nonconvex constraint ${\rm rank}(X)=1$. We then get a upper bound on the optimal value of 
\eqref{eq:maxcut} by solving
\begin{equation}
\begin{array}{ll} \text{maximise} & \displaystyle  {\rm tra}(WX) \\
\text{subject to} & X \succeq 0, \rule[3.5mm]{0mm}{0mm}\\
&(\forall \ell\in{\cal L}) \,\,\, X_{\ell\ell} = 1,\rule[3.5mm]{0mm}{0mm}
\end{array}
\label{eq:SDP}
\end{equation}
which is called the SDP relaxation of the original nonconvex QCQP. There exist several strong results on the upper bound of the gap between the optimal solution and the solution obtained from semidefinite relaxations to NP-hard problems  \cite{goe:95,kar:94}.

We can solve \eqref{eq:mctra} in a distributed fashion. Assume each node $i$ has only knowledge of the weights $(w_{ij})_{j\in{\cal N}_i}$. 
To decouple the node dependencies, we introduce local copies $X_i$ of $X$, and require that $(\forall \{i,j\}\in\E) \,\, X_i = X_j$. With this, the distributed version of \eqref{eq:SDP}  can be expressed as
\begin{equation}
\begin{array}{lll} \text{maximise} & \displaystyle \sum_{i\in\V} {\rm tra}(W_iX_i) \\
\text{subject to} & (\forall i\in\V) &\hspace{-4mm}\left\{ \!\!\begin{array}{l} X_i \succeq 0,\rule[3.5mm]{0mm}{0mm} \\  
(\forall \ell\in{\cal L}) \,\,(X_i)_{\ell\ell} = 1, \rule[3.5mm]{0mm}{0mm} \end{array}\right. \\
&(\forall \{i,j\}\in\E) &\hspace{-4mm}X_i=X_j, \rule[3.5mm]{0mm}{0mm}
\end{array}
\label{eq:SDP_d}
\end{equation}
where $\forall i\in\V, \forall m\in\V, \forall j\in\V $, we have 
\begin{align}
&(W_i)_{mj} =  \left\{ \begin{array}{ll}\displaystyle \sum_{j\in{\cal N}_i} w_{mj}, & \text{if } m=j=i,\\ -\frac{1}{2} w_{mj}, & \text{if }  m=i, j\in{\cal N}_i, \rule[4mm]{0mm}{0mm} \\ 
-\frac{1}{2} w_{mj}, & \text{if }  j=i, m\in{\cal N}_i, \rule[4mm]{0mm}{0mm} \\
0, & \text{otherwise.} \rule[4mm]{0mm}{0mm}\end{array} \right. \nonumber
\end{align}
Hence, $(\forall i\in\V) \, W_i \succeq 0$ and $W_i$ only contains valuesf $(w_{ij})_{j\in{\cal N}_i}$ which are locally known to node $i$. Moreover, $\sum_{i \in\V} W_i = W$, so that problem \eqref{eq:SDP_d} is equivalent to problem \eqref{eq:SDP}. Hence, by
Proposition~\ref{prop:conv} (or Proposition~\ref{prop:sconv} in case of stochastic updates), the distributed solution will converge to the centralised one.

Problem \eqref{eq:SDP_d} is of the form of our prototypical problem, where we can express the constraint $(X_i)_{\ell\ell}=1$  as $h_\ell(X_i) = {\rm tra}(E_\ell X_i) = 1$, where $E_\ell$ is a matrix which is zero everywhere except for entry $(\ell,\ell)$, which is 1. 
As explained in Section~\ref{sec:operator}, these constraints can be easily implemented as correction terms to the solution of the unconstraint updates of the primal variables. Since ${\rm tra}(E_\ell E_m) = \delta_{\ell m}$, \eqref{eq:gamma} becomes $(\tilde{\gamma}_i^{(k)})_\ell = (\tilde{X}_i^{(k)})_{\ell\ell}-1$, and the updates \eqref{eq:xic} become
\[
(\forall i\in\V) \;\;\;
\left\{ \begin{array}{ll} \tilde{X}_i^{(k)} = (W_i - C_i^*Z^{(k)})/(c\,{\rm deg}\,i), \\
 X_i^{(k)} =\tilde{X}_i^{(k)} - \displaystyle \sum_{\ell\in{\cal L}} \big((\tilde{X}_i^{(k)})_{\ell\ell}-1\big)E_\ell. \rule[6mm]{0mm}{0mm}
 \end{array}
 \right.
\]
Hence, the constraints $(\forall \ell\in{\cal L}) \,\,(X_i)_{\ell\ell} = 1$ can be implemented 
by simply setting, at each iteration $k$, the diagonal elements of $\tilde{X}_i^{(k)}$ equal to 1.

Note that due to the relaxation, we have in general ${\rm rank}\, X^* \neq 1$, and we have to extract an approximate solution $\tilde{x}$ from $X^*$ that is feasible. To do so, we can use randomised rounding \cite{goe:95,par:18}. Consider the optimisation problem
\[
\begin{array}{ll} \text{maximise} & \displaystyle \mathbb{E}_{\xi\sim{\cal N}(0,\Sigma)} \left(\xi^TW\xi\right) \\
\text{subject to} & (\forall \ell\in{\cal L}) \,\,\,  \mathbb{E}_{\xi\sim{\cal N}(0,\Sigma)} \left(\xi^2_{\ell}\right) = 1,\rule[4mm]{0mm}{0mm}
\end{array}
\]
where $\mathbb{E}_{x\sim P} f(x) = \int f(x)  {\rm d}P(x)$. This can be equivalently expressed as
 \[
\begin{array}{ll} \text{maximise} & \displaystyle {\rm tra}( W\Sigma)\\
\text{subject to} & \Sigma \succeq 0, \rule[3.5mm]{0mm}{0mm}\\
& (\forall \ell\in{\cal L}) \,\,\, \Sigma_{\ell\ell} = 1,\rule[3.5mm]{0mm}{0mm}
\end{array}
\]
which is of the form \eqref{eq:SDP} with $X=\Sigma$ and leads to the following rounding procedure: starting from a solution
$X^*$ from \eqref{eq:SDP}, we randomly draw several samples $\xi^{(j)} \sim {\cal N}(0,X^*)$, set $\tilde{x}^{(j)} = {\rm sign}\,(\xi^{(j)})$,  and keep the $\tilde{x}^{(j)}$ with the largest objective value. Even though this approach computes good approximate solutions with, in some cases, hard bounds on their suboptimality, it is not suitable for a distributed implementation since after the generation of the random samples we need to compute the objective function for selecting the best candidate which requires knowledge of $W$.
Alternatively, a simple heuristic for finding a good partition is to solve the SDPs above  and approximate $X^*$ by the best rank-1 approximation. That is, we set $\hat{X}= \lambda_1 q_1 q_1^T$, where $\lambda_1$  
is the largest eigenvalue of $X^*$ and $q_1$ the corresponding eigenvector, and we may define $\tilde{x} = {\rm sign}\,(q_1)$ as our candidate solution to \eqref{eq:maxcut}.

Fig.~\ref{fig:mc} shows convergence results for the max-cut problem \eqref{eq:SDP_d} over the graph depicted in Fig.~\ref{fig:graph_25nodes}. The (symmetric) weights are drawn uniformly at random from the set $\{0,\ldots,9\}$ and the step size parameter $c$ was set to $c = 1$. Fig.~\ref{fig:mc}(a) shows convergence result of the primal variable, averaged over all entries of $X$ and all 25 nodes. The objective value of the solution using the best rank-1 approximation was in this case 398, while the objective value using the randomised rounding method, generating 500 samples, was 401.  Fig.~\ref{fig:mc}(b) shows the distribution of 500 objective values for points sampled using randomised rounding.
\begin{figure}[t]
\centering
\includegraphics[width=.47\textwidth]{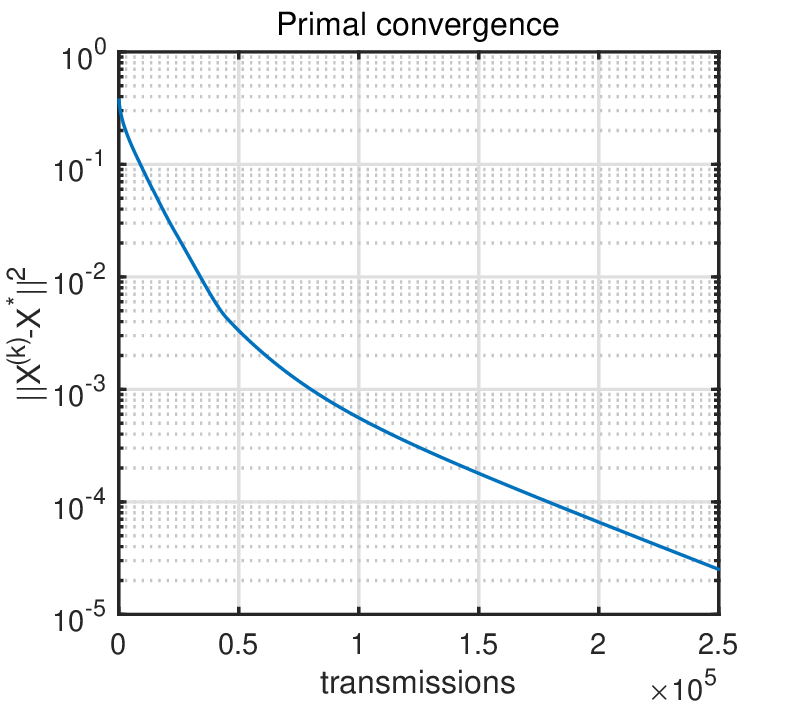}
\includegraphics[width=.51\textwidth]{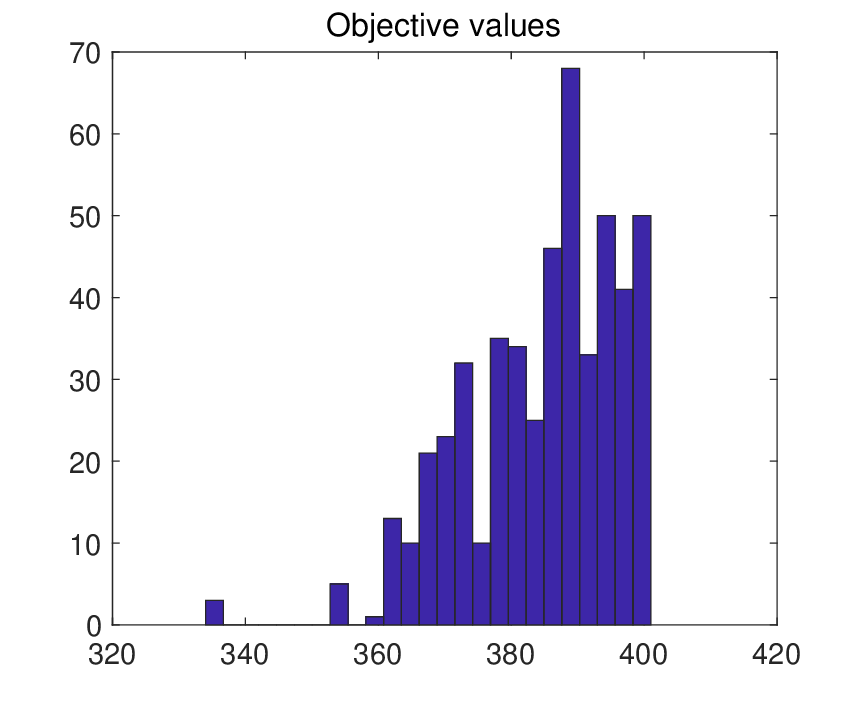}\\
\hspace{3mm} $a) \hspace{63mm} b)$ \vspace{-.3\baselineskip}
\caption{Convergence results for GPDMM for the max-cut problem.  Fig.\ (a) shows   convergence result of the primal variable, averaged over all entries of $X$ and all nodes, while (b) shows the distribution of 500 objective values for points sampled using randomised rounding.}
\label{fig:mc}
\end{figure}
In order to determine the distribution of the difference between the objective value obtained from the best rank-1 approximation and the one from randomised rounding, we run $10^2$ Monte-Carlo simulations. 
Fig.~\ref{fig:mcsim} shows the histogram of the differences. In 55$\%$ of the cases the difference is at most 1, and in 90$\%$ of the cases the difference is at most 10. The mean value of the objective values is 422.7 and the mean value of the differences  is 3.8.
\begin{figure}[t]
\centering
\includegraphics[width=.49\textwidth]{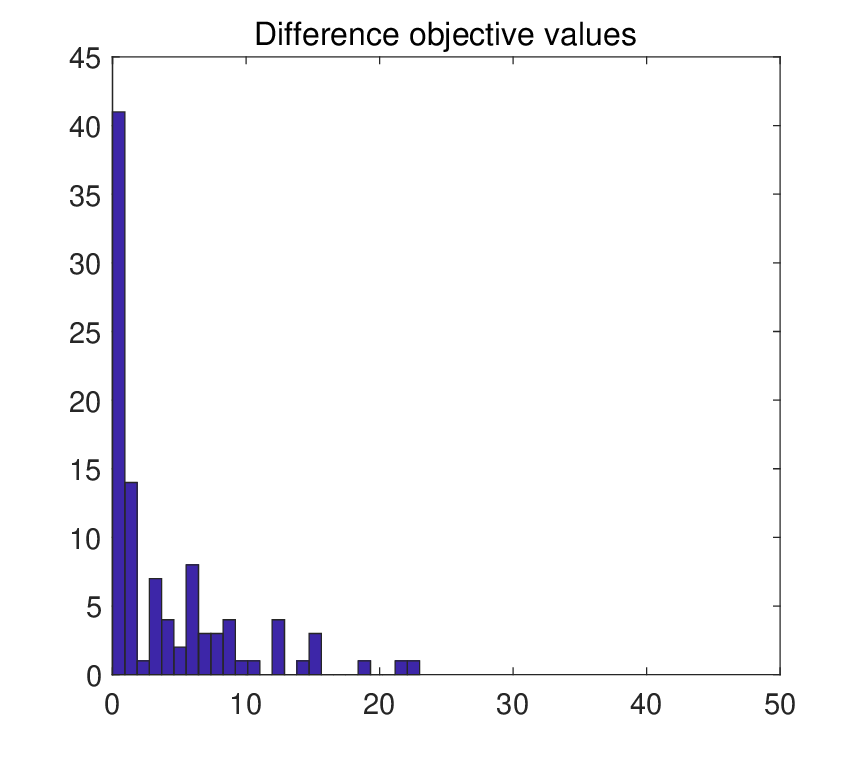}
\caption{Histogram of difference objective values.}
\label{fig:mcsim}
\end{figure}

\section{Conclusions}
\label{sec:conclusion}

In this paper we considered the problem of distributed nonlinear optimisation of a separable convex cost function over a graph subject to cone constraints. We demonstrated the extension of the primal-dual method of multipliers (PDMM) to include cone constraints. We derived update equations, using convex analysis, monotone operator theory and fixed-point theory, by applying the Peaceman-Rachford splitting algorithm to the monotonic inclusion related to the lifted dual problem. The cone constraints are implemented by a reflection method in the lifted dual domain where auxiliary variables are reflected with respect to the intersection of the polar cone and a subspace relating the dual and lifted dual domain. We derived convergence results for both synchronous and stochastic update schemes and demonstrated an application of the proposed algorithm to the maximum cut problem.  
%



%
\end{document}